\documentclass[preprint]{elsarticle} 

\usepackage{epstopdf}
\usepackage[caption=false]{subfig}

\usepackage{setspace}

\usepackage[T1]{fontenc}        
\usepackage{geometry}
\usepackage[utf8x]{inputenc}    
\usepackage{ucs}                
\usepackage{amsmath}            
\usepackage{amsfonts}
\usepackage{amssymb}
\usepackage{amsthm} 
\usepackage{thmtools}
\usepackage{ifthen}
\usepackage{color}
\usepackage{graphicx} 
\usepackage{textcomp}            
\usepackage{palatino}            
\usepackage{eulervm}            
\usepackage{varioref}            
\usepackage{numcompress}\biboptions{authoryear}
\usepackage{cleveref}             

\usepackage{float}
\usepackage{multirow}
\usepackage{rotating}
\usepackage{subfig}
\usepackage{xr} 
\usepackage{setspace} 
\usepackage{tabularx}
\usepackage{booktabs}
\usepackage{hhline}

\newcommand{\R}{{\mathbb{R}}}         
\newcommand{\E}{{\mathbb{E}}}

\newcommand{\bay}{\begin{array}}
\newcommand{\eay}{\end{array}}

\newcommand{\bqa}{\begin{eqnarray*}}
\newcommand{\eqa}{\end{eqnarray*}}

\newcommand{\bee}{\begin{eqnarray*}}
\newcommand{\eee}{\end{eqnarray*}}

\newcommand{\bea}{\begin{eqnarray*}}
\newcommand{\eea}{\end{eqnarray*}}

\newcommand{\bqan}{\begin{eqnarray}}
\newcommand{\eqan}{\end{eqnarray}}

\newcommand{\be}{\begin{eqnarray}}
\newcommand{\ee}{\end{eqnarray}}

\newcommand{\bit}{\begin{itemize}}
\newcommand{\eit}{\end{itemize}}

\newcommand{\ben}{\begin{enumerate}}
\newcommand{\een}{\end{enumerate}}

\newcommand{\beq}{\begin{equation}}
\newcommand{\eeq}{\end{equation}}

\newcommand{\bdes}{\begin{description}}
\newcommand{\edes}{\end{description}}

\newcommand{\btb}{\begin{tabular}}
\newcommand{\etb}{\end{tabular}}

\newcommand{\bcen}{\begin{center}}
\newcommand{\ecen}{\end{center}}

\newcommand{\bmp}{\begin{minipage}}
\newcommand{\emp}{\end{minipage}}

\newcommand{\Cov}{\operatorname{{\it Cov}}}

\newcommand{\Var}{\operatorname{{\it Var}}}

\newcommand{\tr}{\operatorname{tr}}

\newcommand{\im}{\operatorname{\it Im}}

\newcommand{\rank}{\operatorname{\it rank}}

\newcommand{\vh}{\boldsymbol{h}}

\newcommand{\vr}{\boldsymbol{r}}

\newcommand{\vx}{\boldsymbol{x}}
\newcommand{\vy}{\boldsymbol{y}}

\newcommand{\vA}{\boldsymbol{A}}

\newcommand{\vH}{\boldsymbol{H}}
\newcommand{\vI}{\boldsymbol{I}}
\newcommand{\vJ}{\boldsymbol{J}}

\newcommand{\vL}{\boldsymbol{L}}

\newcommand{\vP}{\boldsymbol{P}}
\newcommand{\vQ}{\boldsymbol{Q}}

\newcommand{\vT}{\boldsymbol{T}}

\newcommand{\vV}{\boldsymbol{V}}

\newcommand{\vX}{\boldsymbol{X}}

\newcommand{\vbeta}{\boldsymbol{\beta}}

\newcommand{\veta}{\boldsymbol{\eta}}

\newcommand{\vmu}{\boldsymbol{\mu}}

\newcommand{\vSigma}{\boldsymbol{\Sigma}}

\newcommand{\vtheta}{\boldsymbol{\theta}}

\newcommand{\veins}{{\bf 1}}
\newcommand{\vnull}{{\bf 0}}

\newcommand{\Cline}[2]{
  \noalign{\global\setlength{\arrayrulewidth}{#1}}\cline{#2}%
  \noalign{\global\setlength{\arrayrulewidth}{0.4pt}  }}

\DeclareMathOperator{\vech}{vech}

\newtheoremstyle{Test1}
  {2 \baselineskip}
  {1.5 \baselineskip}
  {\itshape}
  {-0.0ex}
  {\fontfamily{ppl}\fontseries{l}\fontshape{n}}
  {:}
  {\newline}
   {}

\theoremstyle{Test1}
\newtheorem{Sa}{Theorem}[section]
\newtheorem{theorem}{Theorem}[section]
\newtheorem{re}[Sa]{Remark}
 
\newtheorem{Le}[Sa]{Lemma}

\newcolumntype{x}[1]{!{\centering\arraybackslash\vrule width #1}}

\makeatletter
\renewenvironment{proof}[1][\proofname]{\par
  \pushQED{\qed}%
  \fontfamily{ppl}\fontseries{m}\fontshape{it} \topsep6\p@\@plus6\p@\relax
  \trivlist
  \item[\hskip\labelsep
        \bfseries
    #1\@addpunct{:}]\ignorespaces
}{%
  \popQED\endtrivlist\@endpefalse
}
\makeatother

\begin{document}
\title{Choice of the hypothesis matrix for using the Anova-type-statistic}
\author[1]{Paavo Sattler\corref{cor1}%
}
\ead{paavo.sattler@tu-dortmund.de}
\affiliation[1]{Institute for Mathematical Statistics and Industrial Applications, Faculty of Statistics, Technical
University of Dortmund, Joseph-von-Fraunhofer-Strasse 2-4, 44221 Dortmund, Germany}
\author[2]{Manuel Rosenbaum}
\affiliation[2]{Institute of Statistics, Ulm University, Helmholtzstr. 20, 89081 Ulm, Germany.}
\cortext[cor1]{Corresponding author}
\begin{abstract}
Initially developed in \cite{brDeMu:1997}, the Anova-type-statistic (ATS) is one of the most used quadratic forms for testing multivariate hypotheses for a variety of different parameter vectors $\vtheta\in\R^d$. 
Such tests can be based on several versions of ATS, and in most settings, they are preferable over those based on other quadratic forms, for example, the Wald-type-statistic (WTS). However, the same null hypothesis $\vH\vtheta=\vy$ can be expressed by a multitude of hypothesis matrices $\vH\in\R^{m\times d}$ and corresponding vectors $\vy\in\R^m$, which leads to different values of the test statistic, as it can be seen in simple examples. Since this can entail differing test decisions, it remains to investigate under which conditions certain tests using different hypothesis matrices coincide.
Here, the dimensions of the different hypothesis matrices can be substantially different, which has exceptional potential to save computation effort.

In this manuscript, we show that for the Anova-type-statistic and some versions thereof, it is possible for each hypothesis $\vH\vtheta=\vy$ to construct a companion matrix $\vL$ with a minimal number of rows, which not only tests the same hypothesis but also always yields the same test decisions. This can substantially reduce computation time, which is investigated in several conducted simulations. 

\end{abstract}

\maketitle
Formulating the appropriate hypothesis is a central part of statistical inference. However, right now, there is hardly any material on how to choose a suitable hypothesis matrix to check hypotheses regarding an unknown $d$-dimensional parameter. To formulate the hypothesis of interest, a so-called hypothesis matrix $\vH\in \R^{m\times d}$ is used together with an associated vector $\vy\in \R^m$, where $m\leq d$, which yields $\mathcal{H}_0: \vH\vtheta=\vy$.
Thereby, the parameter vector $\vtheta$ can here be any useful quantity with at least two components, which can also result from a multiple group comparison of common univariate quantities. Examples would be the expectation vector $\vmu$, in one group or more group designs as in \cite{pEB} or \cite{sattler2018}, but also vectorized correlation matrices $\vr$ (see, e.g. \cite{sattler2023correlation}) or coefficient vectors of regression models (as in\cite{Fahrmeir2013RegressionMM}) are feasible. In survival analysis, also the restricted mean survival time is a viable parameter \cite{munko2024RMST}. Even for functional data (like in \cite{munko2023functional}), the vector $\veta$ of mean functions can be investigated, demonstrating the versatility of this model for hypotheses. Unfortunately, the hypothesis matrix $\vH$ together with the corresponding vector $\vy$ is not unique, as is evident, since for each nonsingular matrix $\vQ$ we know $\vH\vtheta=\vy \Leftrightarrow \vQ\vH\vtheta=\vQ\vy$. 
But not even the dimension $m$ is unique, as it can be demonstrated with
\begin{equation}
 \vH_1=  \begin{pmatrix}
1 &-1 &0\\
0 & 1 &-1\\
1 & 0 &-1
\end{pmatrix},
 \quad 
\quad\vH_2=\begin{pmatrix}
1 &-1 &0\\
0 & 1 &-1\\
\end{pmatrix}\label{examples}\end{equation}
which both allow to check the equality of the components of a three-dimensional vector through $\vH_1\vtheta=\vnull_3$ resp. $\vH_2\vtheta=\vnull_2$. This fact indicates the potential for substantial reductions of computation time and complexity of the related tests.

Despite the common and broad application of test statistics based on quadratic forms, only for multiple contrast tests clear standards for the choice of the hypothesis matrix has been provided (see, e.g. \cite{rubarth2022} and \cite{pöhlmann2024}).
In \cite{sattler2023hypothesis}, this issue was paid attention to for the first time, where the focus was on a special quadratic form, the Wald-type-statistic (WTS). Therein, it was shown that each representation of the null hypothesis yields the same value of the test statistic, which allows the use of matrices that are in some sense preferable, for example, computationally or regarding their interpretability. Although the WTS is a very useful test statistic, especially for permutation approaches, it comes with restrictions on the covariance matrix of the underlying data, and moreover, is known to be very liberal for small and medium sample sizes. Therefore, we will focus on another, even wider used quadratic form, the so-called Anova-type-statistic (ATS) and some versions thereof. While in \cite{sattler2023hypothesis} the ATS was only dealt with marginally, we will now investigate under which circumstances multiple ATS coincide for different hypotheses matrices and thereby derive recommendations for an appropriate choice of those matrices.

\section{Main results}
The ATS is a comparably accessible and broadly used quadratic form. Basically for a given matrix $\vH \in \R^{m \times d}$ and a vector $\vy \in \R^m$, it is the function mapping a vector $\vx \in \R^d$ to $\R$ through
\[ATS(\vx,\vH,\vy)= (\vH\vx-\vy)^\top (\vH\vx-\vy).\]
Based on the ATS, it is straightforward to construct a test for the multivariate hypothesis $\mathcal{H}_0:\vH\vtheta = \vy$, as shortly sketched in the following. Given a data set $\vX$, consisting of one or multiple groups, let $\vT(\vX)$ be based on an appropriate estimator for the parameter vector of interest $\vtheta$. For example, in case of the expectation vector $\vtheta=\vmu$, a common test statistic is given through $\vT(\vX)=\sqrt{N}\ \overline \vX$, where $N$ is the sample size. Now, applying the ATS for the test statistic $\vT(\vX)$ usually leads to meaningful tests. 
Also, the standardized version of the ATS, which is known for its robustness against change of unit, is frequently considered and given by
\[ATS_s(\vx,\vH,\vy)=ATS(\vx,\vH,\vy)/\tr(\vH\vSigma\vH^\top),\]
with $\vSigma=\Cov(\vT(\vX))\geq 0$. There exist many other versions beyond, like
\[ATS_F(\vx,\vH,\vy)=ATS_s(\vx,\vH,\vy)\cdot \frac{[\tr(\vH\vSigma\vH^\top)]^2}{\tr(\vH\vSigma\vH^\top\vH\vSigma\vH^\top)} \]
which, under some well-known assumptions, allows to approximate the asymptotic distribution through a $F$- distribution.

We want to point out that for each hypothesis $\mathcal{H}_0:\vH\vtheta=\vy$ it is simple to construct $\widetilde \vH \in\R^{\ell\times d}$ and $\widetilde \vy\in\R^\ell$ such that $\mathcal{H}_0:\widetilde\vH\vtheta=\widetilde\vy$ expresses the same hypothesis. For example, we obtain the equivalent hypothesis when multiplying with a positive diagonal matrix, which means weighting the rows of the hypothesis matrix.

However, note that the corresponding ATS based on these hypothesis matrices does not necessarily yield the same value. 
Hence, it remains to be specified under which conditions two different hypothesis matrices with their corresponding vectors lead to the identical ATS. 
Since both are functions, equality here of course means the same value for each $\vx\in\R^d$.
For a special type of hypotheses, the answer to this question is comparably simple and stated in the next theorem.

\begin{theorem}\label{Theorem1}
    Let $\vH\vtheta=\vnull_m$ describe a linear solution set with $\vH\in \R^{m\times d}$  and let $\vL\in \R^{\ell\times d}$ be a matrix fulfilling $\vL^\top\vL=\vH^\top\vH$.  Then   $\vH\vtheta=\vnull_m$ and $\vL\vtheta=\vnull_\ell$ describe the same null hypothesis, and it further holds
    \begin{itemize}
        \item[1)]$ATS(\vx,\vL,\vnull_\ell)=ATS(\vx,\vH,\vnull_m)\quad $ $\forall \vx\in \R^d$,
        \item[2)]$ATS_s(\vx,\vL,\vnull_\ell)=ATS_s(\vx,\vH,\vnull_m)\quad $ $\forall \vx\in \R^d$,
        \item[3)]$ATS_F(\vx,\vL,\vnull_\ell)=ATS_F(\vx,\vH,\vnull_m)\quad $ $\forall \vx\in \R^d$.
    \end{itemize}
Moreover, $\vL^\top\vL=\vH^\top\vH$ is not only a sufficient but also a necessary condition for 1), 2) and 3).
\end{theorem}

Now we are able to verify if, for two given hypothesis matrices $\vH$ and $\vL$, the corresponding Anova-type-statistics have the same value and thus validate the equivalence of the two null hypotheses expressed by $\vH\vtheta=\vnull_m$ and  $\vL\vtheta=\vnull_\ell$.
However, usually only $\vH$ is given, and hence, the existence and the construction of an appropriate matrix $\vL$ remains to be addressed. Fortunately, in this case, the existence of such matrices is guaranteed, while at the same time, it is possible to specify the minimal dimension $\ell$.

\begin{Le}\label{Lemma1}
  Let  $\vH\in \R^{m\times d}$ with $r=\rank(\vH)$, then there exists a matrix $\vL\in \R^{r\times d}$ with $\vH^\top \vH=\vL^\top\vL$.
\end{Le}

We will call this matrix $\vL$ a compact root of $\vH^\top \vH$ since in the case of $r=d$, which means full row rank of $\vH$, it accords with the classical matrix root but in general is smaller.

It is important to mention that, in general, this compact root is not unique, which can be seen with a simple example. For $\gamma$ in $[0,2\pi]$, consider
\[\vH=\begin{pmatrix}
    1 &0 &0\\
    0 & 1 &0\\
    0 & 0&0
\end{pmatrix}\quad \text{and}\quad  \vL_{\gamma}=\begin{pmatrix}
    \sin(\gamma)&-\cos(\gamma)&0\\
    \cos(\gamma)&\sin(\gamma)&0\\
    
\end{pmatrix} .\]
Then $\vH^\top \vH = \vL_\gamma^\top \vL_\gamma$ holds for each $\gamma$.

Nevertheless, since the result is based on the spectral decomposition, for which algorithms in the usual computation languages exist, it is easy to calculate one of such roots with adequate software (an example for the computation environment R can be found in the supplement).

Further, in many areas, like split-plot designs, hypothesis matrices are written as Kronecker-product of two smaller matrices, the subplot matrix $\vH_S$ and the wholeplot matrix $\vH_W$, see for example \cite{sattler2018} or \cite{sattler2021}. 
In this case, it is sufficient to calculate the compact root for the factors and build their Kronecker-product instead of calculating the compact root of the product, as $\vL_S^\top\vL_S=\vH_S^\top \vH_S$ and $\vL_W^\top\vL_W=\vH_W^\top \vH_W$ yield $(\vL_W\otimes \vL_S)^\top (\vL_W\otimes \vL_S)=(\vH_W\otimes \vH_S)^\top (\vH_W\otimes \vH_S)$.

Now, if $r<m$, by \Cref{Lemma1}, each compact root of $\vH^\top\vH$ is a smaller matrix $\vL$, still expressing the same null hypothesis and also yielding the same test statistic for all mentioned versions of the ATS. 
This allows a more compact expression of the null hypothesis and is especially attractive regarding the computation time.
In particular, for high-dimensional settings or bootstrap and permutation approaches, where the quadratic form is calculated for a large number of bootstrap or permutation runs, choosing a parsimonious hypothesis matrix and, hence, reducing the computation time for each Anova-type statistic can be very beneficial.

\begin{re}\label{Unique}
Note, that for two matrices $\vH_1\in\R^{m_1\times d}$ and $\vH_2\in\R^{m_2\times d}$ expressing the same null hypothesis through $\vH_1\vtheta=\vnull_{m_1}$ and $\vH_2\vtheta=\vnull_{m_2}$ respectively, the corresponding $\vL_1$ and $\vL_2$ are not necessary equal. This is inconvenient since we are interested in a unique procedure to represent the null hypothesis. However, we can bypass this difficulty by using the unique projection matrix $\vP\in\R^{d \times d} $. It is a well known fact that $\vP\vtheta=\vnull_d\Leftrightarrow\vH_1\vtheta=\vnull_{m_1}$ and $\vP\vtheta=\vnull_d\Leftrightarrow\vH_2\vtheta=\vnull_{m_2}$ and moreover, both matrices $\vH_1$ and $\vH_2$ lead to the same matrix $\vP$.
Based on this matrix $\vP$, we can calculate a compact root for $\vP^\top \vP=\vP$, which doesn't depend on the concrete hypothesis matrix with which we started. Since with \Cref{Theorem1} all compact roots of $\vP$ lead to the same hypothesis and the identical value of the test statistics, the concrete choice of the matrix $\vL$ has no influence.
\end{re}

As a matter of fact \Cref{Theorem1} in general only holds for $ATS(\vx,\vH,\vnull_m)$ and not for arbitrary $ATS(\vx,\vH,\vy)$ with $\vy\in\R^m$. Even for simple examples and $\widetilde \vy\in \R^\ell$ it gets clear that $ATS(\vx,\vH,\vy)\neq ATS(\vx,\vL,\widetilde \vy)$ although $\vH^\top\vH=\vL^\top\vL$. The following theorem specifies the prerequisites that ensure $ATS(\vx,\vH,\vy)=ATS(\vx,\vL,\widetilde \vy)$. This usually requires conditions on the hypothesis matrices, the corresponding vectors and the combination of both.

\begin{theorem}\label{Theorem3}
Let $\vH\vtheta=\vy$ and $\vL\vtheta=\tilde \vy$ be two non-empty solution sets with hypothesis matrices $\vH\in \R^{m\times d}$ and $\vL\in\R^{\ell\times d}$, a parameter vector $\vtheta\in\R^d$ and corresponding vectors $\vy\in \R^{m}$ and $\tilde\vy\in \R^\ell$. Then it holds
  \begin{itemize}
     \item[a)] If $a\cdot \vL^\top\vL=\vH^\top\vH$ and $a\cdot \vL^\top \tilde \vy=\vH^\top\vy$  is fulfilled for an $a\in \R\setminus\{0\} $, it holds that $\vH\theta=\vy$ and $\vL\theta=\tilde \vy$ describe the same hypothesis.
           \item[b)]  $ATS_s(\vx,\vH,\vy)=ATS_s(\vx,\vL,\tilde \vy)$ if and only if $a\cdot \vL^\top\vL=\vH^\top\vH$, $ {a}\cdot\vL^\top \tilde \vy=\vH^\top\vy$ and $||\vy||=\sqrt{a}\cdot||\tilde \vy|| $ are fulfilled for an $a\in \R\setminus\{0\} $.
            \item[c)]  $ATS_F(\vx,\vH,\vy)=ATS_F(\vx,\vL,\tilde \vy)$ if and only if $a\cdot \vL^\top\vL=\vH^\top\vH$, $ {a}\cdot\vL^\top \tilde \vy=\vH^\top\vy$ and $||\vy||= \sqrt{a}\cdot||\tilde \vy||$ are fulfilled for an $a\in \R\setminus\{0\} $.
       \end{itemize}
\end{theorem}

These conditions are independent of the data, and hence, they can be easily verified in practice.
\begin{re}
 Part b) of \Cref{Theorem3} can be formulated similarly but slightly weaker for $ATS$, as can be seen in the supplement. 
\end{re}

\begin{re}\label{Multiplerow}
 In  Corollary 3 of \cite{sattler2023hypothesis}, it was already shown that for a matrix $\vH\in \R^{m\times d}$ and corresponding vector $\vy\in\R^{m}$, having rows which are multiples of other ones, there exists a matrix $\widetilde \vH\in\R^{\ell \times d}$ and a vector $\widetilde \vy\in\R^{\ell}$ with $\ell\leq m$, fulfilling
 $ATS(\vx,\vH,\vy)=ATS(\vx,\widetilde \vH, \widetilde \vy)$ and $ATS_s(\vx,\vH,\vy)=ATS_s(\vx,\widetilde \vH,\widetilde\vy)$. From the construction of $\widetilde \vH$ from $\vH$, it directly follows that $\vH^\top\vH=\widetilde \vH^\top\widetilde \vH$, $\vH^\top\vy=\widetilde\vH^\top\widetilde \vy$ and $||\vy||=||\widetilde\vy||$ which set both results in relation.
 \end{re}

\begin{Le}\label{Lemma3}
  Let  $\vH\in \R^{m\times d}$ and $\vy\in\R^m$  formulating a hypothesis $\vH\vtheta=\vy$ with non-empty solution set. Then there exists a matrix $\vL\in \R^{r\times d}$ and a vector $\widetilde \vy\in\R^r$, with $r=\rank(\vH)$, fulfilling the conditions of \Cref{Theorem3}.
 
\end{Le}

Now \Cref{Lemma2} from the supplement shows a potential procedure to find for given hypothesis $\vH\vtheta=\vy$  a matrix and a vector fulfilling the conditions to apply \Cref{Theorem3}. First, a compact root $\vL$ of $\vH^\top\vH$ is calculated, and subsequent for this matrix, the system $\vL^\top\widetilde \vy=\vH^\top\vy$ is solved to find a sufficient $\widetilde \vy$. Here, neither the compact root nor the solution of the linear equation system has to be unique. 
If, then $||\widetilde \vy||=||\vy||$ the conditions are fulfilled. Otherwise we define $\sqrt{a}=||\vy||/||\widetilde\vy||$ and obtain the wanted matrix through $\sqrt{a}\cdot \vL$ and the corresponding vector through $\sqrt{a}\cdot \widetilde\vy$.\\\\
In contrast to \Cref{Unique}, for $\vy\neq \vnull_d$, there is no general guideline on how to choose the hypothesis matrix $\vH$, so it is possible to use different matrices. Now, for each $\vH$, one can find an adequate $\vL$, but of course, two hypothesis matrices $\vH_1$ and $\vH_2$ for the same hypothesis can lead to different $\vL_1$ and $\vL_2$  and different values, unless they fulfil also the conditions of \Cref{Theorem3}.
Therefore, each test of a multivariate hypothesis needs to specify the hypothesis matrix used.

\begin{re}
It is evident that the ATS for different hypothesis matrices may differ not only in their test decision but also regarding their power, as it can be seen for $\vH_1$ and $\vH_2$ from \eqref{examples}. This leads to the idea of valuing different hypothesis matrices regarding their power under some alternatives. Although this seems a reasonable approach, it is not feasible since the asymptotic distribution of the ATS usually is a weighted sum of independent $\chi_1^2$ distributed variables. Since the weights strongly depend on the hypothesis matrix, this also holds for the distribution under the null hypothesis, which prohibits comparability.\\
Therefore, it is recommended that the hypothesis matrix be formulated carefully and in a way that all alternatives are considered uniformly. Although this can result in extensive hypothesis matrices at first, the corresponding hypothesis matrix based on the compact root nonetheless provides the minimal number of rows.
\end{re}

The applicability of our results is not limited to the three presented versions of the ATS nor their usage as a test statistic. For example, in \cite{pEB} and \cite{sattler2018}, another kind of ATS is used for a high-dimensional setting. Several trace estimators are defined within this procedure, which are U-statistics based on the ATS. 
All of these estimators are based on a large number of such quadratic forms, which can be implemented using a compact root. This is extremely attractive in combination with the high dimension of the underlying vectors.

\section*{Simulation}
To investigate the influence of the chosen hypothesis matrix on the computation time, in this section, we consider three different settings and subsequently different hypotheses
\begin{itemize}
\item[A)]Consider independent $q$ dimensional observation vectors   $\vX_{ik}$, $k \in \{1,...,n_i\}$ from two groups which are identically distributed within these groups and have sample size $n_i$, $i \in \{1,2\}$. With $\E(\vX_{ij})=\vmu_i\in\R^d$, $\Cov(\vX_{ij})=\vV_i>0$ in context of repeated measure a usual null hypothesis is the hypothesis of no group effect for $\vmu=(\vmu_1^\top,\vmu_2^\top)^\top$. This can be expressed through $(\vP_2 \otimes \vJ_q)\vmu=\vnull_{2q}$, with $\vP_d=\vI_d-\vJ_d/d$, where $\vJ_d$ is a $d\times d$ matrix containing only 1's and $\vI_d$ is the identity with dimension $d$. 

\item[B)]
In the setting of A) but for three groups, we are interested in the hypothesis of equal expectation vectors of the groups, given through
$(\vP_3\otimes \vI_q)\vmu=\vnull_{3q}$.

\item[C)]
Now consider independent observations $\vX_1,...,\vX_{n}$ with covariance matrix $\vV\in\R^{p\times p}$. To investigate the hypothesis  $\mathcal{H}_0: \tr(\vV)=\gamma$ in \cite{sattler2022} the upper triangular vectorisation 
\[ \vech(\vV)=(v_{11},...,v_{1p},v_{22},...,v_{2p},...,v_{pp})^\top \in\R^{d}, \] $d=p(p+1)/2$ was used. This hypothesis can be expressed through $(\vh_p\cdot \vh_p^\top)\vech(\vV)=\gamma \cdot \veins_d$, with $\vh_p:=(1,\vnull_{p-1}^\top, 1, \vnull_{p-2}^\top, \dots, 1,0,1)^\top\in \R^{d}$. 

\end{itemize}
Setting A) and C) were also investigated in \cite{sattler2023hypothesis} for the WTS.
Our focus is on the calculation of the ATS, so the above hypothesis serves only as motivation, and we generate $\vX_{A)}\sim\mathcal{N}_{2q}(\vnull_{2q},\vV_{2q})$, $\vX_{B)}\sim\mathcal{N}_{3q}(\vnull_{3q},\vV_{3q})$ and $\vX_{C)}\sim\mathcal{N}_{p}(\veins_{p},\vV_{p})$.

Here $\vV_{d}$ denotes a d-dimensional compound symmetry matrix given through $\vV_{d}=\vI_d+\veins_d\veins_d^\top$, and we consider $q=(5,10,20,50,100,200)$ and $p=(5,10,15,20,25,30)$.

To obtain more reliable results, the time required for 5.000 calculations of the quadratic form is measured, which is a usual number of bootstrap steps or permutations. The computations were executed using the \textsc{R}-computing environment version 4.3.3 \cite{R}  on servers with two AMD EPYC 7453 28-Core processors and 256 GB DDR4 memory with Debian GNU/Linux 12. The required time in seconds for the $ATS_s$ is displayed in \Cref{tab:Zeit3}, while the values for the $ATS$ can be found in the appendix.\\\\

  \begin{table}[ht]
\centering

      \begin{tabular}{x{1.5pt}lx{1.5pt}lx{1.5pt}r|r|r|r|r|rx{1.5pt}}
       \specialrule{1.5pt}{0pt}{0pt}
   \multicolumn{1}{x{1.5pt}lx{1.5pt}}{}
&\multicolumn{7}{ cx{1.5pt}}{calculation time in seconds}
\\ \specialrule{1.5pt}{0pt}{0pt}
 \multirow{ 3}{*}{A)}&d(q)   &10(5)&20(10)&40(20)&100(50)&200(100)  &400(200) \\  
       \Cline{1.0pt}{2-8}
  &  $ATS_s(\vP_2 \otimes \vJ_q,\vnull_{2q})$& 0.863 & 1.033 & 1.660 & 10.519 & 45.278 & 314.871 \\   \Cline{0.2pt}{2-8}
   &   $ATS_s(\vL_{A)},0)$ & 0.739 & 0.775 & 0.819 & 1.043 & 1.837 & 4.709 \\   
       \specialrule{1.5pt}{0pt}{0pt}
       
       \multicolumn{8}{x{1.5pt} cx{1.5pt}}{}\\
\specialrule{1.5pt}{0pt}{0pt}
\multirow{ 3}{*}{B)} &   d(q)   &15(5)&30(10)&60(20)&150(50)&300(100)  &600(200) \\
     \Cline{1.0pt}{2-8}
  &     $ATS_s(\vP_3\otimes \vI_q,\vnull_{3q})$ & 0.963 & 1.313 & 2.944 & 22.241 & 160.721 & 1076.300 \\   \Cline{0.2pt}{2-8}
  &     $ATS_s(\vL_{B)},\vnull_{2q})$ & 0.914 & 1.152 & 2.133 & 15.342 & 87.141 & 577.770 \\      
      \specialrule{1.5pt}{0pt}{0pt}

          \multicolumn{8}{x{1.5pt} cx{1.5pt}}{}\\
\specialrule{1.5pt}{0pt}{0pt}
\multirow{ 3}{*}{C)} &   d(p)& 15(5)&55(10)&120(15)&210(20)&325(25)&465(30)\\
    \Cline{1.0pt}{2-8}
    &   $ATS_s(\vh_p\cdot \vh_p^\top,\gamma\cdot\veins_d)$ &0.935 & 2.835 & 12.153 & 53.132 & 179.819 & 526.618 \\   \Cline{0.2pt}{2-8}
    &   $ATS_s(\vL_{C)},\gamma)$ &  0.742 & 0.856 & 1.156 & 1.994 & 3.542 & 6.159 \\          
      \specialrule{1.5pt}{0pt}{0pt}   
      \end{tabular}

\caption{ Average computation time in seconds of  $ATS_s$ based on different hypothesis matrices with different dimensions.}
  \label{tab:Zeit3}
\end{table}
It is clear that the settings strongly differ in the ratio between the number of rows $d$ of the classical hypothesis matrix, and $\ell$ minimal number of required rows. This gets clear in setting C), which is an extreme case and therefore leads already for a 10-dimensional observation vector $\vX$ to savings of two-thirds of the computation time, and for large dimensions to even much larger savings, nesting between 98-99\%. Also considering part A), for groups with dimension $p=20$, more than 50\% of the computation time could be saved by using the smaller matrix, which grows up to 98\% for larger dimensions about 200.\\
In B) it is easy to see that $\ell/d=2/3$ is the largest value of all hypotheses considered here (and therefore the smallest reduction of rows), but savings are nevertheless substantial. Although for the low dimensional case $p=5$, the total computation time is only in the area of a few seconds already, for 5-dimensional observation vectors, more than 5\% can be saved, which increases to more than 30\% for observation vectors with $q\geq 20$. Since this is a widespread hypothesis for comparing three groups, it shows the great benefit of our results. 
{Moreover, even if a time reduction of about 5\% for low dimensions $d$ seems small, it has to be taken into account that the result of the tests is not influenced, which means that no price has to be paid for this reduction.\\
Further, we want to emphasize that the demonstrated reduction is only for the $ATS$. In other situations with $ATS_s$, also the covariance matrix $\Var(\vT(\vX))=\vSigma$ has to be estimated to afterwards calculate the trace of $\vH\vSigma\vH^\top$. This can be done more efficiently by estimating $\Var(\vH\vT(\vX))$ where again, the usage of $\vL$ has a large potential.

\section*{Conclusion}
Together with the WTS, the ATS is among the most common quadratic forms for testing multivariate hypotheses. Hereby, in contrast to the WTS, the chosen hypothesis matrix usually affects the value of the test statistic and, therefore, also the p-value and the test decision. In the presented work, we investigated in detail under which conditions two hypothesis matrices lead to the identical ATS while studying the equality of the corresponding null hypotheses.
In the interest of computation time, we also proposed a procedure for each hypothesis matrix to construct a companion matrix with a minimal number of rows while maintaining the hypothesis and yielding the identical value of the test statistic.
This can reduce the calculation effort substantially and is exceptional in its generality. Moreover, it can be used for several versions of the ATS as well as estimators based on such quadratic forms and provides even further possibilities. In future work, a general guideline on how to choose the hypothesis matrix in case of $\vy\neq \vnull_d$ would complete our results.
\section*{Acknowledgements}
The work of Paavo Sattler was funded by
Deutsche Forschungsgemeinschaft Grant/Award Number DFG PA 2409/3-2. Paavo Sattler's research was
conducted while visiting Heidelberg University as a guest scientist. He is grateful
for the wonderful research environment and hospitality and especially the support of Ricardo Blum for the presented work.

\bibliography{Literatur}
\bibliographystyle{elsarticle-num}

\section*{Appendix}
\subsection{Proofs and further results}
\begin{proof}[Proof of \Cref{Theorem1}]
   
     For the different versions of the ATS, it is easy to calculate the equalities through\\\\
     $ATS(\vx,\vL,\vnull_\ell)= (\vL\vx)^\top (\vL\vx)=\vx^\top (\vL^\top \vL)\vx=\vx^\top (\vH^\top \vH)\vx=ATS(\vx,\vH,\vnull_m)$,\\\\
$ATS_s(\vx,\vL,\vnull_\ell)= ATS(\vx,\vL,\vnull_\ell)/\tr(\vL\vSigma\vL^\top)=ATS(\vx,\vL,\vnull_\ell)/\tr(\vSigma\vL^\top\vL)=ATS(\vx,\vH,\vnull_m)/\tr(\vSigma\vH^\top\vH)=ATS_s(\vx,\vH,\vnull_m)$,\\\\
$ATS_F(\vx,\vL,\vnull_\ell)=ATS_s(\vx,\vL,\vnull_\ell)\cdot \frac{[\tr(\vL\vSigma\vL^\top)]^2}{\tr(\vL\vSigma\vL^\top\vL\vSigma\vL^\top)}=ATS_s(\vx,\vH,\vnull_m)\cdot \frac{[\tr(\vSigma\vL\vL^\top)]^2}{\tr(\vSigma\vL\vL^\top\vSigma\vL\vL^\top)}=ATS_s(\vx,\vH,\vnull_m)\cdot \frac{[\tr(\vSigma\vH\vH^\top)]^2}{\tr(\vSigma\vH\vH^\top\vSigma\vH\vH^\top)}=ATS_F(\vx,\vH,\vnull_m) $.\\\\
So it remains to show, that this condition is also necessary. We know that $ATS(\vx,\vL,\vnull_\ell)$ is a polynomial function in $\vx$ and is therefore differentiable. Now if $ATS(\vx,\vL,\vnull_\ell)=ATS(\vx,\vH,\vnull_m)$ also the corresponding Hessian matrix are equal. They are easy to calculate, which leads to 
$\vL^\top\vL=\vH^\top\vH$.

Now  we compare the two hypotheses:\\

 Let $\vH\vtheta=\vnull_m$, then it is clear \[ \vH^\top\vH\vtheta=\vnull_d \Rightarrow\vL^\top \vL\vtheta=\vnull_d \Rightarrow \vtheta^\top\vL^\top(\vL\vtheta)=0 \Leftrightarrow ||\vL\vtheta||^2=0\Leftrightarrow \vL\vtheta=\vnull_\ell.\] Identically, if 
    $\vL\vtheta=\vnull_\ell$ it follows
     \[\vL^\top\vL\vtheta=\vnull_d \Rightarrow\vH^\top \vH\vtheta=\vnull_d \Rightarrow \vtheta^\top\vH^\top\vH\vtheta=0\Leftrightarrow ||\vH\vtheta||^2=0\Leftrightarrow\vH\vtheta=\vnull_m.\]
     Therefore, both solution sets are equal, and so also both null hypotheses are equal.
\end{proof}

\begin{proof}[Proof of \Cref{Lemma1}]
It is clear that $\vA:=\vH^\top \vH$ is symmetric and positive semidefinit. Moreover, it is known that $\rank(\vA)=\rank(\vH)=r$. Now it is a well-known result (see e.g. Theorem 4.5 in \cite{schott2017}),  that there exists $\vL\in\R^{r\times d}$ with $\vL^\top\vL=\vA=\vH^\top\vH$.
\end{proof}

\begin{theorem}\label{Theorem2}
Let $\vH\vtheta=\vy$ and $\vL\vtheta=\widetilde \vy$ be two non-empty solution sets with hypothesis matrices $\vH\in \R^{m\times d}$ and $\vL\in\R^{\ell\times d}$, a parameter vector $\vtheta\in\R^d$ and corresponding vectors $\vy\in \R^{m}$ and $\widetilde\vy\in \R^\ell$. Then it holds
 \begin{itemize}
    
     \item[a)] $ATS(\vx,\vH,\vy)=ATS(\vx,\vL,\widetilde \vy)-(||\widetilde \vy||^2-||\vy||^2)$ if and only if $\vL^\top\vL=\vH^\top\vH$ and $ \vL^\top \widetilde \vy=\vH^\top\vy$ are fulfilled      
      \item[b)] $ATS(\vx,\vH,\vy)=ATS(\vx,\vL,\widetilde \vy)$ if and only if $\vL^\top\vL=\vH^\top\vH$, $ \vL^\top \widetilde \vy=\vH^\top\vy$ and $||\vy||=||\widetilde \vy||$ are fulfilled.
      
 \end{itemize}

\end{theorem}

\begin{proof}[Proof of \Cref{Theorem2}]
\noindent
\begin{itemize}
  
\item[a)]
 If $ATS(\vx,\vH,\vy)=ATS(\vx,\vL,\widetilde \vy)-(||\widetilde \vy||^2-||\vy||^2)=0$ and we consider the difference of the two quadratic forms we get
 \[ATS(\vx,\vH,\vy)-ATS(\vx,\vL,\widetilde \vy)=(\vx^\top \vH^\top\vH\vx-\vx^\top\vL^\top\vL\vx)+2\vx^\top(\vL^\top\widetilde{\vy}-\vH^\top{\vy})-(\widetilde \vy^\top\widetilde \vy-\vy^\top\vy)\]
 and therefore
 \[(\vx^\top \vH^\top\vH\vx-\vx^\top\vL^\top\vL\vx)+2\vx^\top(\vL^\top\widetilde{\vy}-\vH^\top\vy)=0\quad \forall \vx\in\R^d\]
for all $\vx\in\R^d$. Therefore it also holds for the corresponding Hessian matrix, resulting in $\vH^\top\vH-\vL^\top\vL=\vnull_{d\times d}$. Based on this if we consider the gradient we get $(\vL^\top\widetilde{\vy}-\vH^\top{\vy})=\vnull_d$ which fulfills the proof.
\item[b)]
This result follows directly from the previous one.
  \end{itemize}
\end{proof}
Now, if only $\vH^\top\vy=\vL^\top\widetilde \vy$ and $\vL^\top\vL=\vH^\top\vH$ are fulfilled, the difference between both statistics is a known value and also not depends on the data. So, the test statistics are only shifted, and their quantiles are also shifted. This leads to identical test decisions, and therefore the biggest issue with different values of the test decisions did not occur.

\begin{Le}\label{Lemma2}
  Let  $\vH\in \R^{m\times d},\vy\in\R^m$ with $r=\rank(\vH)$ formulating a hypothesis $\vH\vtheta=\vy$ with non-empty solution set. Then there exists a matrix $\vL\in \R^{r\times d}$ and a vector $\widetilde \vy\in\R^r$ fulfilling $\vH^\top \vH=\vL^\top\vL$ and $\vH^\top\vy=\vL^\top\widetilde \vy$. 
\end{Le}

In particular, \Cref{Lemma2} yields that for each hypothesis $\vH\vtheta=\vy$, the statements formulated in parts a) of \Cref{Theorem3} and \Cref{Theorem2} respectively hold. Thus, we can rephrase each hypothesis $\vH\vtheta=\vy$ through a smaller matrix $\vL$, still maintaining the same hypothesis and only shifting the value of the corresponding ATS by a known amount, dependent on $\vy$ and $\widetilde\vy$.

\begin{proof}[Proof of \Cref{Lemma2}]
Based on \Cref{Lemma1} it follows that $\vL$, the compact root of $\vH^\top\vH$, fulfills $\vH^\top\vH=\vL^\top\vL$. So it remains to show that for each $\vy\in \R^m$ there exists a $\widetilde\vy\in\R^r$ fulfilling $\vH^\top\vy=\vL^\top\widetilde \vy$.\\
From \Cref{Theorem1} we know $\ker(\vH)=\ker(\vL)$ and therefore $\ker(\vH)^\perp=\ker(\vL)^\perp$. Now, since $\vH^\top\vy \in \im(\vH^\top)$ together with $\im(\vH^\top)=\ker(\vH)^\perp=\ker(\vL)^\perp=\im(\vL^\top)$, it follows $\vH^\top\vy \in \im(\vL^\top)$. 
Therefore, at least one  $\widetilde\vy\in\R^r$ with $\vH^\top\vy=\vL^\top\widetilde \vy$ exists, which completes the proof.
\end{proof}
Here, since the matrix and the corresponding vector are not unique, some potential $\widetilde\vy$ may fulfil $||\vy||=||\widetilde\vy||$ while others do not.

\begin{proof}[Proof of \Cref{Theorem3}]
\noindent
\begin{itemize}\item[a)]

  Let $\vH\vtheta=\vy$ then it is clear \[ \vH^\top\vH\vtheta=\vH^\top\vy \Rightarrow a\cdot \vL^\top \vL\vtheta=a\cdot \vL^\top\widetilde \vy \Rightarrow \vL^\top(\vL\vtheta-\widetilde \vy)=\vnull_d \Leftrightarrow \vL\vtheta-\widetilde \vy\in \ker(\vL^\top) .\]
  Now, we can use that the kernel of $\vL^\top$ is the same as the orthogonal complement of the image of $\vL$, denoted by $\im(\vL)^\perp$.
  At the same time, we know that because of the non-empty solution set, $\widetilde\vy \in \im(\vL)$, obviously it also holds $\vL\vtheta\in \im(\vL)$ and therefore also $\vL\vtheta-\widetilde \vy\in\im(\vL)$. Since we are in a finite-dimensional vector space, the only vector which can be of $\im(\vL)$ and $\im(\vL)^\perp$ is the zero vector and hence $\vL\vtheta=\widetilde \vy$. The other direction can be proven identically, which completes the proof.

\item[b)]
Define $\widetilde \vH=\frac{\vH}{\sqrt{\tr(\vSigma\vH^\top\vH)}}$, $\widetilde \vL=\frac{\vL}{\sqrt{\tr(\vSigma\vL^\top\vL)}}$   as well as $\vbeta=\frac{\vy}{\sqrt{\tr(\vSigma\vH^\top\vH)}}$ and $\widetilde \vbeta=\frac{\widetilde\vy}{\sqrt{\tr(\vSigma\vL^\top\vL)}}$. Then it holds
$ATS_s(\vx,\vH,\vy)=ATS(\vx,\widetilde \vH,\vbeta)$ and $ATS_s(\vx,\vL,\widetilde \vy)=ATS(\vx,\widetilde \vL,\widetilde\vbeta)$. Now we can use \Cref{{Theorem2}} and get 
$ATS(\vx,\widetilde \vH,\vbeta)=ATS(\vx,\widetilde \vL,\widetilde\vbeta)$ if an only if $\widetilde \vH^\top \widetilde \vH=\widetilde \vL^\top \widetilde \vL$, $\widetilde \vH^\top\vbeta=\widetilde \vL^\top \widetilde \vbeta$ and $||\vbeta||=||\widetilde\vbeta||$. This leads to

\[\frac{\tr(\vSigma\vH^\top\vH)}
{\tr(\vSigma\vL^\top\vL)}\cdot 
\vL^\top\vL=\vH^\top\vH\]
and therefore $ \vL^\top\vL$ and $\vH^\top\vH$ have to be proportional. For this reason, such an $a\in\R$ has to exist and in the same way we get  $||\vy||=\sqrt{a}\cdot ||\widetilde \vy||$. From $\widetilde \vH^\top\vbeta=\widetilde \vL^\top \widetilde \vbeta$ it follows $ {a}\cdot\vL^\top \widetilde \vy=\vH^\top\vy$.

\item[c)]The proof follows identical to part b), with appropriate $\widetilde \vH$ and $\widetilde \vL$, resp. $\vbeta$ and $\widetilde\vbeta$. 
\end{itemize}
\end{proof}

\begin{proof}[Proof of \Cref{Lemma3}]

From the proof of \Cref{Lemma2},  we know how to find $\vL$ and $\widetilde\vy$ to fulfill $\vL^\top\vL=\vH^\top\vH$ and $\vL^\top\widetilde \vy=\vH^\top\vy$. 
Now the matrix $\sqrt{a}\cdot \vL$ and the vector $\sqrt{a}\cdot \widetilde\vy$ of course fullfilling $a\cdot\vL^\top\vL=\vH^\top\vH$ and $a\cdot\vL^\top\widetilde \vy=\vH^\top\vy$.
And with $\sqrt{a}\cdot ||\widetilde\vy||=||\vy||$ it also fullfilles the third condition of \Cref{Theorem3}, which completes the proof.
\end{proof}
\subsection{Further simulations}
Since the matrix multiplication and calculating the trace for the standardization requires some calculation time, it is unsurprising that all durations here are clearly smaller than those for the $ATS_s$.
Although the percentage of saved computation time for lower dimensions is here slightly lower than in \Cref{tab:Zeit3}, for higher ones, it could be analogous as for parts A) and B) but also reversed as in part C).
All together, it shows one more time the potential of this reduction of hypothesis matrices rows.

 \begin{table}[ht]
\centering
      \begin{tabular}{x{1.5pt}lx{1.5pt}lx{1.5pt}r|r|r|r|r|rx{1.5pt}}
       \specialrule{1.5pt}{0pt}{0pt}
  \multicolumn{2}{x{1.5pt}lx{1.5pt}}{}
&\multicolumn{6}{ cx{1.5pt}}{calculation time in seconds}
\\ \specialrule{1.5pt}{0pt}{0pt}
 &d(q)   &10(5)&20(10)&40(20)&100(50)&200(100)  &400(200) 
      \\  
       \Cline{1.1pt}{2-8}
A)   & $ATS(\vP_2 \otimes \vJ_q,\vnull_{2q})$& 0.256 & 0.283 & 0.318 & 0.529 & 1.218 & 4.105 \\   \Cline{0.2pt}{2-8}
      &$ATS(\vL_{A)},0)$ &  0.245 & 0.258 & 0.269 & 0.286 & 0.291 & 0.317 \\ 
       \specialrule{1.5pt}{0pt}{0pt}
       
      \multicolumn{8}{x{1.5pt} cx{1.5pt}}{}\\
     
\specialrule{1.5pt}{0pt}{0pt}
\multirow{ 3}{*}{B)} &   d(q)   &15(5)&30(10)&60(20)&150(50)&300(100)  &600(200) \\
     \Cline{1.1pt}{2-8}
&       $ATS(\vP_3\otimes \vI_q,\vnull_{3q})$ & 0.263 & 0.291 & 0.376 & 0.829 & 2.507 & 8.681 \\  \Cline{0.2pt}{2-8}
   &    $ATS(\vL_{B)},\vnull_{2d})$ &  0.272 & 0.283 & 0.345 & 0.655 & 1.772 & 5.938 \\       
      \specialrule{1.5pt}{0pt}{0pt}

        \multicolumn{8}{x{1.5pt} cx{1.5pt}}{}\\
\specialrule{1.5pt}{0pt}{0pt}
\multirow{ 3}{*}{C)} &   d(p)& 15(5)&55(10)&120(15)&210(20)&325(25)&465(30)\\
    \Cline{1.pt}{2-8}
 &      $ATS(\vh_p\cdot \vh_p^\top,\gamma\cdot\veins_d)$ & 0.263 & 0.356 & 0.638 & 1.392 & 2.930 & 5.412 \\  \Cline{0.2pt}{2-8}
 &      $ATS(\vL_{C)},\gamma)$ &  0.255 & 0.270 & 0.293 & 0.306 & 0.318 & 0.334 \\           
      \specialrule{1.5pt}{0pt}{0pt}   
      \end{tabular}

\caption{ Average computation time in seconds of  $ATS$ based on different hypothesis matrices with different dimensions.}
  \label{tab:Zeit4}
\end{table}
\subsection{Code for the compact root in R}
\begin{verbatim}
 MSrootcompact<- function(X){
  if(length(X) == 1){
    MSroot <- matrix(sqrt(X),1,1)
  }
  else{
    r=rankMatrix(X)[1]
    SVD <- svd(X)
    MSroot <-  sqrt(diag(SVD$d[1:r],r,r))%*%t((SVD$u)[,1:r])
  }
  return(MSroojat)
}   
\end{verbatim}
\end{document}